\documentclass[submission,copyright,noderivs]{eptcs}
\providecommand{\event}{NCMA 2023} % Name of the event you are submitting to

\usepackage{graphicx} 
\usepackage{mathrsfs}
\usepackage{amssymb}
\usepackage{amsmath}
\usepackage{amsthm}
\usepackage{pifont}
\usepackage[dvipsnames]{xcolor}

\theoremstyle{plain}
\newtheorem{theorem}{Theorem}
\newtheorem{proposition}[theorem]{Proposition}
\newtheorem{lemma}[theorem]{Lemma}
\newtheorem{corollary}[theorem]{Corollary}

\theoremstyle{definition}
\newtheorem{example}[theorem]{Example}

\newcommand{\lfam}{\mathscr{L}}

\newcommand{\rtf}{\lfam_{\subtext{rt}}}

\newcommand{\reg}{\textrm{REG}}

\newcommand{\dcfl}{\textrm{DCFL}}
\newcommand{\dcsl}{\textrm{DCSL}}

\newcommand{\twsda}{\textsf{twsDA}}
\newcommand{\twsdnea}{\textsf{twsDNEA}}
\newcommand{\twsdca}{\textsf{twsDCA}}
\newcommand{\dsa}{\textsf{DSA}}
\newcommand{\dnesa}{\textsf{DNESA}}
\newcommand{\dcsa}{\textsf{DCSA}}

\newcommand{\type}{\textsf{TYPE}}
\newcommand{\direct}{\textsf{DIRECT}}

\DeclareMathOperator{\pdpop}{\texttt{pop}}
\DeclareMathOperator{\pdpush}{\texttt{push}}

\newcommand{\rightend}{\mathord{\vartriangleleft}}
\newcommand{\dollar}{\texttt{\$}}
\newcommand{\cent}{\texttt{\textcent}}
\newcommand{\border}{\texttt{\#}}

\newcommand{\eoe}{\hspace*{\fill} $\blacksquare$\smallskip}

\newcommand{\subtext}[1]{\textnormal{\scriptsize #1}}

\newcommand{\cyes}{\textcolor{PineGreen}{\ding{51}}}%
\newcommand{\cno}{\textcolor{red}{\ding{55}}}%

\title{Deterministic Real-Time Tree-Walking-Storage Automata}

\author{Martin Kutrib
\institute{%
  Institut f\"ur Informatik, Universit\"at Giessen\\
  Arndtstr.~2, 35392 Giessen, Germany}
\email{kutrib@informatik.uni-giessen.de}
\and
Uwe Meyer
\institute{%
    Technische Hochschule Mittelhessen\\
    Wiesenstr.~14, 35390 Giessen, Germany}
 \email{uwe.meyer@mni.thm.de}
}
 
\def\titlerunning{Deterministic Real-Time Tree-Walking-Storage Automata}
\def\authorrunning{M.~Kutrib, U.~Meyer}

\begin{document}

\maketitle

\begin{abstract}
We study deterministic tree-walking-storage automata, which
are finite-state devices equipped with a tree-like storage. 
These automata are generalized stack automata, where the linear 
stack storage is replaced by a non-linear tree-like stack.
Therefore, tree-walking-storage automata have the
ability to explore the interior of the tree storage 
without altering the contents, with the possible moves of the tree pointer
corresponding to those of tree-walking automata.
In addition, a tree-walking-storage automaton can append (push) non-existent
descendants to a tree node and remove (pop) leaves from the tree.
Here we are particularly considering the capacities of deterministic 
tree-walking-storage automata working in real time. It is shown that
even the non-erasing variant can accept rather complicated unary languages
as, for example, the language of words whose lengths are powers of two, or
the language of words whose lengths are Fibonacci numbers. Comparing the
computational capacities with automata from the classical automata hierarchy,
we derive that the families of languages accepted by real-time
deterministic (non-erasing) tree-walking-storage automata is located between
the regular and the deterministic context-sensitive languages. There is
a context-free language that is not accepted by any real-time
deterministic tree-walking-storage automaton. On the other hand,
these devices accept a unary language in non-erasing mode that cannot be
accepted by any classical stack automaton, even in erasing mode and arbitrary time.
Basic closure properties of the induced families of languages are
shown. In particular, we consider 
Boolean operations (complementation, union, intersection) and AFL operations 
(union, intersection with regular languages, homomorphism, inverse homomorphism, 
concatenation, iteration). It turns out that the two families in question have 
the same properties and, in particular, share all but one of these closure 
properties with the important family of deterministic context-free languages.
%
%\keywords{Deterministic one-way stack automata \and Tree-walking automata \and Tree-stacks
%\and Real time \and Computational capacity \and Closure properties}
\end{abstract}

\section{Introduction}

Stack automata were introduced in \cite{Ginsburg:1967:SAC} as a
theoretical model motivated by compiler theory, and the implementation of
recursive procedures with parameters. Their computational power lies 
between that of pushdown automata and Turing machines. 
Basically, a stack automaton is a finite-state device equipped with a
generalization of a pushdown store. In addition to be able to push
or pop at the top of the pushdown store, a stack automaton can move 
its storage head (stack pointer) \emph{inside} the stack to read  
stack symbols, but without altering the contents. 
In this way, it is possible to read but not to change the stored information.
Over the years, stack automata have aroused great interest and have 
been studied in different variants. Apart from distinguishing deterministic
and nondeterministic computations, the original two-way input reading variant
has been restricted to one-way~\cite{Ginsburg:1967:owsa:art}.
Further investigated restrictions
concern the usage of the stack storage. A stack automaton is said to be
\emph{non-erasing} if no symbol may be popped from the
stack~\cite{hopcroft:1967:nesa},
and it is \emph{checking} if it cannot push any symbols once the 
stack pointer has moved into the stack~\cite{greibach:1969:caaowsl}.
While the early studies 
of stack automata have extensively been done in relation with AFL theory as well
as time and space 
complexity~\cite{gurari:1982:sasa,hopcroft:1968:dsaqo,ibarra:1971:csttcc,lange:2010:nopcdowsl,shamir:1974:cscfpgapcl},
more recent papers consider the computational power gained in generalizations by
allowing the input head to jump~\cite{kosarajyu:1974:owsawj}, allowing
multiple input heads, multiple stacks~\cite{ibarra:2021:gocsach}, and 
multiple reversal-bounded counters~\cite{ibarra:2018:vocsaoudp}.
The stack size required to accept a language by stack automata has been
considered as well~\cite{ibarra:2021:scosam}.
In~\cite{kutrib:2017:tdpda} the property of working input-driven has been 
imposed to stack automata, and their capacities as transducer are studied
in~\cite{bensch:2017:dst}.

All these models have in common that their storage structures are linear. 
Therefore, it is a natural idea to generalize stack automata by replacing
the stack storage by some non-linear data structure. 
In~\cite{kutrib:2023:twsa:proc} tree-walking-storage automata have been
introduced, which are essentially stack automata
with a tree-like stack. As for classical stack automata,
tree-walking-storage automata have the additional
ability to move the storage head (here tree pointer) inside the 
tree without altering the contents. The possible moves of the tree pointer
correspond to those of tree walking automata.
In this way, it is possible to read but not to change the stored information.
In addition, a tree-walking-storage automaton can append (push) a non-existent
descendant to a tree node and remove (pop) a leaf from the tree. 
A main focus in~\cite{kutrib:2023:twsa:proc} is on the comparisons of 
the different variants of tree-walking-storage
automata as well as on the comparisons with classical stack automata.
It turned out that the checking variant is no more powerful than 
classical checking stack automata. In particular it is shown
that in the case of unlimited time 
deterministic tree-walking-storage automata are as powerful as Turing
machines. This result suggested to consider
time constraints for deterministic tree-walking-storage automata.
The computational capacities of polynomial-time
non-erasing tree-walking-storage automata and non-erasing stack automata
are separated. Moreover, it is shown that 
non-erasing tree-walking-storage and tree-walking-storage automata
are equally powerful.

Here we continue the study of tree-walking-storage automata by imposing a
very strict time limit. We consider the minimal time to solve non-trivial
problems, that is, we consider real-time computations. This natural limitation
has been investigated from the early beginnings of complexity theory.
Already before the seminal paper~\cite{Hartmanis:1965:cca}, Rabin
considered computations such that if the problem (the input data) consists
of $n$ symbols then the computation must be performed in $n$ basic steps, one step
per input symbol~\cite{rabin:1963:rtc}.

Before we turn to our main results and the organization of the paper, we
briefly mention different approaches to introduce tree-like stacks.
So-called pushdown tree automata~\cite{Guessarian:1983:pta}
extend the usual string pushdown automata by allowing trees instead of
strings in both the input and the stack. So, these machines accept trees
and may not explore the interior of the stack. Essentially, this model
has been adapted to string inputs and tree-stacks where the so-called
\emph{tree-stack automaton} can
explore the interior of the tree-stack in read-only mode~\cite{golubski:1996:tsa}.
However, in the writing-mode a new tree can be pushed on the stack employing the subtrees of the old
tree-stack, that is, subtrees can be permuted, deleted, or copied. If the root of the tree-stack
is popped, exactly one subtree is left in the store. Another model also
introduced under the name \emph{tree-stack automaton} gave up the bulky way of pushing and popping at
the root of the tree-stack~\cite{denkinger:2016:aacfmcfl:proc}. However, this
model may alter the interior nodes of the tree-stack. Therefore, the
tree-stack is actually a non-linear Turing tape.
Therefore, we have chosen the name \emph{tree-walking-storage automaton}, so
as not to have one more model under the name of tree-stack automaton.

The idea of a tree-walking process originates from~\cite{Aho:1971:tocfg}.
A tree-walking automaton is a sequential model that processes input trees. For example, it is known that
deterministic tree-walking automata are strictly weaker than nondeterministic
ones~\cite{Bojanczyk:2006:twacbd} and that even nondeterministic tree-walking automata cannot accept
all regular tree languages~\cite{Bojanczyk:2008:twadnrarl}.

The paper is organized as follows.
The definition of the models and an illustrating example are given in
Section~\ref{sec:prelim}.
Section~\ref{sec:comp-cap}
is devoted to compare the computational capacity
of real-time deterministic tree-walking-storage automata with
some classical types of acceptors.
It is shown that the possibility to create
tree-storages of certain types in real time can be utilized to
accept further, even unary, languages by
real-time deterministic, even non-erasing, tree-walking-storage automata.
To this end, the non-semilinear unary language of the words whose lengths
are double Fibonacci numbers is used as a witness.

Then, a technique for disproving that languages are 
accepted is established for real-time tree-walking-storage automata.
The technique is based on equivalence
classes which are induced by formal languages. If some language
induces a number of equivalence classes which exceeds the number 
of classes distinguishable by a certain device, then the
language is not accepted by that device.
Applying these results, we show that there is a context-free language
which is not accepted by any tree-walking-storage automaton in real time.
For the comparison with classical deterministic one-way stack automata
we show that the unary language $\{\,a^{n^3}\mid n\geq 0\,\}$
is a real-time tree-walking-storage automaton language. It is known
from~\cite{ogden:1969:itsl} that this language  
is not accepted by any classical deterministic one-way stack automaton.
Finally, in Section~\ref{sec:closure} some basic closure properties
of the language families in question are derived.
It turns out that the two families in question have the same properties and,
in particular, share all but one of these closure properties with the important family
of deterministic context-free languages.
In particular, we consider 
Boolean operations (complementation, union, intersection) and AFL operations 
(union, intersection with regular languages, homomorphism, inverse homomorphism, 
concatenation, iteration). 
The results are summarized in Table~\ref{tab:closure} at the end of the section.

\section{Definitions and Preliminaries}\label{sec:prelim}

Let $\Sigma^*$ denote the \emph{set of all words} over the finite alphabet
$\Sigma$.
The \emph{empty word} is denoted by $\lambda$, and
$\Sigma^+ = \Sigma^* \setminus \{\lambda\}$. The set of words of length 
$n\geq 0$ is denoted by $\Sigma^{n}$.
The \emph{reversal} of a word $w$ is denoted by $w^R$.  For the \emph{length} of~$w$ we
write~$|w|$. 
We use $\subseteq$ for \emph{inclusions} and~$\subset$ for \emph{strict inclusions}.
We write~$|S|$ for the cardinality of a
set~$S$. 
We say that two language families $\mathscr{L}_1$ and~$\mathscr{L}_2$
are \emph{incomparable} if~$\mathscr{L}_1$ is not a subset of~$\mathscr{L}_2$
and vice versa.

A tree-walking-storage automaton is an extension of a classical stack
automaton to a tree storage.  As for classical stack automata,
tree-walking-storage automata have the additional
ability to move the storage head (here tree pointer) inside the 
tree without altering the contents. The possible moves of the tree pointer
correspond to those of tree walking automata.
In this way, it is possible to read but not to change the stored information.
However, a classical stack automaton can push and pop at the top of the
stack. Accordingly, a tree-walking-storage automaton can append (push) a non-existent
descendant to a tree node and remove (pop) a leaf from the tree. 

Here we consider mainly deterministic one-way devices.
The trees in this paper are finite, binary trees whose nodes are labeled by
a finite alphabet $\Gamma$. A $\Gamma$-tree $T$ is represented by a
mapping from a finite, non-empty, prefix-closed subset of $\{l,r\}^*$
to $\Gamma \cup\{\bot\}$, such that $T(w)=\bot$ if and only if $w=\lambda$.
The elements of the domain of $T$ are called \emph{nodes of the tree}. 
Each node of the tree has a \emph{type} from $\type=\{-,l,r\}\times\{-,+\}^2$, where the first component
expresses whether the node \emph{is} the root ($-$), a left descendant ($l$), or a right
descendant ($r$), and the second and third components tell whether the node 
\emph{has} a left and right descendant ($+$), or not ($-$).
A \emph{direction} is an element from $\direct = \{u, s, d_l, d_r\}$, where $u$
  stands for `up', $s$ stands for `stay', $d_l$ stands for `left descendant' 
and $d_r$ for `right descendant'.

A \emph{deterministic tree-walking-storage automaton ($\twsda$)} is a system
\mbox{$M=\langle Q, \Sigma, \Gamma, \delta, q_0, \rightend, \bot, F\rangle$,} where 
$Q$ is the finite set of \emph{internal states},~$\Sigma$ is 
the finite set of \emph{input symbols} not containing the \emph{endmarker}~$\rightend$,
$\Gamma$ is the finite set of \emph{tree symbols},
$q_0 \in Q$ is the \emph{initial state},
$\bot \notin \Gamma$ is the \emph{root symbol},
$F\subseteq Q$ is the set of \emph{accepting states}, and
\begin{multline*}
\delta\colon 
Q \times (\Sigma \cup \{\lambda,\rightend\}) \times \type\times (\Gamma
\cup\{\bot\})
\rightarrow\\
Q \times (\direct \cup \{\pdpop\}\cup\{\,\pdpush(x,d)\mid x\in\Gamma,
d\in\{l,r\}\,\})
\end{multline*}
is the \emph{transition function}.
There must never be a choice of using
an input symbol or of using~$\lambda$ input. So,
it is required that for all $q$ in~$Q$,~$(t_1,t_2,t_3)\in\type$,
and $x$ in~$\Gamma\cup\{\bot\}$:
if $\delta(q,\lambda,(t_1,t_2,t_3),x)$ is defined, then
$\delta(q,a,(t_1,t_2,t_3),x)$
is undefined for all~$a$ in~$\Sigma\cup \{\rightend\}$.

A \emph{configuration} of a $\twsda$
is a quadruple $(q, v, T, P)$,  where $q \in Q$ is the current state, 
\mbox{$v \in \Sigma^*\{\rightend, \lambda\}$} is the unread part of the
input, $T$ is the current $\Gamma$-tree, and $P$ is an element of the domain of $T$,
called the \emph{tree pointer}, that is the current node of $T$.
The \emph{initial configuration} for input $w$ is set to 
$(q_0, w\rightend, T_0, \lambda)$, where  $T_0(\lambda)=\bot$ and $T_0$ is undefined
otherwise.

During the course of its computation, $M$ runs through a
sequence of configurations. 
In a given configuration $(q, v, T, P)$, $M$ is in state $q$, 
reads the first symbol of $v$ or $\lambda$, knows the 
type of the current node $P$, and sees the label~$T(P)$ of 
the current node. 
Then it applies $\delta$ and, thus, enters a new state
and either moves the tree pointer along a direction, removes the current node 
(if it is a leaf) by $\pdpop$, or appends a new descendant to the current
node (if this descendant does not exist) by $\pdpush$.
Here and in the sequel it is understood that~$\delta$ is well defined in the
sense that it will never move the tree pointer to a non-existing node,
will never pop a non-leaf node, and will never push an existing descendant.
This normal form is always available through effective constructions.

One step from a configuration to its
successor configuration is denoted by~$\vdash$, and the reflexive 
and transitive (resp., transitive) closure of~$\vdash$ 
is denoted by~$\vdash^*$ \mbox{(respectively $\vdash^+$).}
Let $q \in Q$, $av\in \Sigma^*\rightend$ with $a \in \Sigma \cup \{
\lambda,\rightend \}$, $T$ be a $\Gamma$-tree,~$P$ be a tree pointer of $T$,
and $(t_1,t_2,t_3)\in\type$ be the type of the current node~$P$.
We set 

\begin{enumerate}
\item $(q, av, T, P)\vdash (q', v, T, P')$
with $P=P'l$ or $P=P'r$,\\
if $t_1\neq -$ and $\delta(q,a,(t_1,t_2,t_3), T(P)) = (q',u)$, 
(move the tree pointer up),
\smallskip
\item $(q, av, T, P)\vdash (q', v, T, P)$,\\
if $\delta(q,a,(t_1,t_2,t_3), T(P)) = (q',s)$, 
(do not move the tree pointer),
\smallskip
\item $(q, av, T, P)\vdash (q', v, T, P')$
with $P'=Pl$,\\
if $t_2= +$ and $\delta(q,a,(t_1,t_2,t_3), T(P)) = (q',d_l)$, 
(move the tree pointer to the left descendant),
\smallskip
\item $(q, av, T, P)\vdash (q', v, T, P')$
with $P'=Pr$,\\
if $t_3= +$ and $\delta(q,a,(t_1,t_2,t_3), T(P)) = (q',d_r)$, 
(move the tree pointer to the right descendant),
\smallskip
\item $(q, av, T, P)\vdash (q', v, T', P')$
with $P=P'l$ or $P=P'r$, $T'(P)$ is undefined and $T'(w)=T(w)$ for $w\neq P$,\\
if $t_2=t_3= -$ and $\delta(q,a,(t_1,t_2,t_3), T(P)) = (q',\pdpop)$, 
(remove the current leaf node, whereby the tree pointer is moved up),
\smallskip
\item $(q, av, T, P)\vdash (q', v, T', P')$
with $P'=Pl$, $T'(Pl)=x$ and $T'(w)=T(w)$ for $w\neq Pl$,\\
if $t_2= -$ and $\delta(q,a,(t_1,t_2,t_3), T(P)) = (q',\pdpush(x,l))$, 
(append a left descendant to the current node, whereby the tree pointer is
moved to the descendant),
\smallskip
\item $(q, av, T, P)\vdash (q', v, T', P')$
with $P'=Pr$, $T'(Pr)=x$ and $T'(w)=T(w)$ for $w\neq Pr$,\\
if $t_3= -$ and $\delta(q,a,(t_1,t_2,t_3), T(P)) = (q',\pdpush(x,r))$, 
(append a right descendant to the current node, whereby the tree pointer is
moved to the descendant).
\end{enumerate}

Figure~\ref{fig:Up-Left-transitions} illustrates the transitions that move the tree pointer up,
respectively to the left descendant.

\begin{figure}[!ht]
    \includegraphics[width=\textwidth]{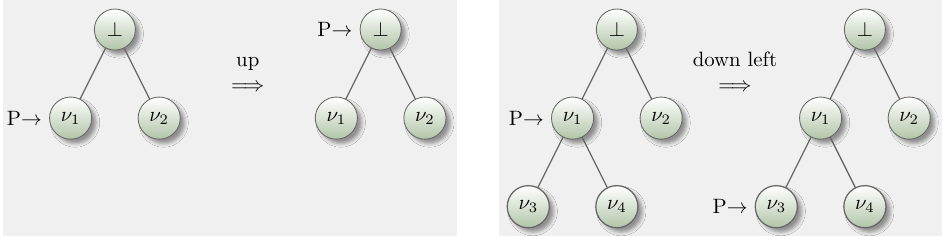}
    \caption{Up and left transitions}
    \label{fig:Up-Left-transitions}
\end{figure}

Figure~\ref{fig:Push-Pop} illustrates the push, respectively the pop transitions.
All remaining transitions are analogous.

\begin{figure}[!ht]
    \includegraphics[width=\textwidth]{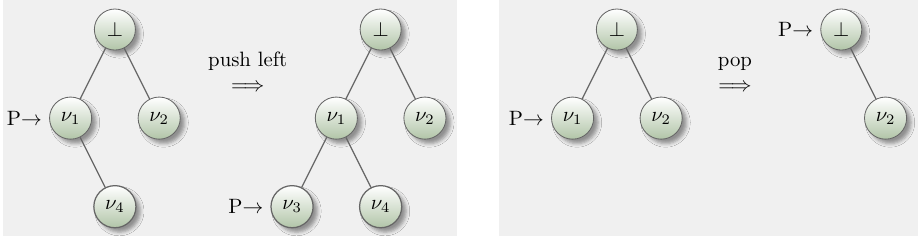}
    \caption{Push left and pop operations}
    \label{fig:Push-Pop}
\end{figure}

So, a classical stack automaton can be seen as a 
tree-walking-storage automaton all of whose right descendants
of the tree-storage are not present.
In accordance with stack automata, a $\twsda$
is said to be \emph{non-erasing} ($\twsdnea$) if it is not allowed to pop
from the tree.

A $\twsdca$ $M$ \emph{halts} if the transition function
is not defined for the current configuration. 
A word $w$ is \emph{accepted} if
the machine halts in an accepting state after having read the input 
$w\rightend$ entirely, otherwise it is
\emph{rejected}. The \emph{language accepted} by $M$ is
$L(M)=\{\, w\in \Sigma^* \mid w \text{ is accepted by } M\,\}$.

A $\twsda$ works in \emph{real time} if its transition function is
undefined for $\lambda$ input. That is, it reads one symbol from the 
input at every time step, thus, halts on input $w$ after at most $|w|+1$ steps.
 
We write $\dsa$ for deterministic one-way stack automata,
$\dnesa$ for the non-erasing, and $\dcsa$ for the checking variant.
The family of languages accepted by a device of type \textsf{X}
is denoted by~$\lfam(\textsf{X})$. We write in particular~$\rtf(\textsf{X})$
if acceptance has to be in real time.

In order to clarify our notion, we continue with an example.

\begin{example}\label{exa:a2n}
The language 
$L_\subtext{expo}=\{\,a^{2^n} \mid n\geq 0\,\}$
is accepted by some $\twsdnea$ in real time.

The basic idea of the construction is to let a $\twsdnea$
successively create tree-storages which are complete binary
trees. To this end, we construct a  
$\twsdnea$ $M=\langle Q, \{a\}, \{\bullet\}, \delta, q_l, \rightend, \bot,
F\rangle$ with state set $Q=\{q_l,q_p,q_d,q_r\}$
that runs in phases. In each phase a complete level is added to 
the complete binary tree. So, at the outset of the computation 
the tree-storage of $M$ forms a complete binary tree of level~$1$, that is
a single node. After the $(\ell-1)$th phase, the
tree-storage of $M$ forms a complete binary tree of level $\ell$, that is,
the tree has $2^\ell-1$ nodes. 
At the beginning and at the end of each phase the tree pointer
is at the root of the tree-storage. For simplicity, we construct $M$ such that
it works on empty input only. Later, it will be extended.

Next, we explain how a level is added when the tree-storage of $M$ forms a 
complete binary tree of level $\ell\geq 1$ and the tree pointer is at the root.

Let a star $*$ as component of the type of the current node in the
tree-walking-storage of~$M$ denote an arbitrary entry 
and $\gamma\in \Gamma\cup\{\bot\}$.
We set:
\begin{enumerate}
\item
$\delta(q_l, \lambda, (*,+,*), \gamma) = (q_l, d_l)$
\smallskip
\item
$\delta(q_l, \lambda, (*,-,-), \gamma) = (q_p,\pdpush(\bullet,l))$
\smallskip
\item
$\delta(q_p, \lambda, (l,-,-), \gamma) = (q_r,u)$
\end{enumerate}

First, state $q_l$ is used to move the tree pointer as far as
possible to the left (Transition 1). The leaf reached is the first node 
that gets descendants. After pushing a left descendant (Transition 2),
$M$ enters state $q_p$ to indicate that the last tree operation
was a push. If the new leaf was pushed as left descendant, the tree pointer
is moved up while state $q_r$ is entered (Transition 3). State $q_r$
indicates that the right subtree of the current node has still to be processed.

\begin{enumerate}\setcounter{enumi}{3}
\item
$\delta(q_r, \lambda, (*,+,-), \gamma) = (q_p,\pdpush(\bullet,r))$
\smallskip
\item
$\delta(q_p, \lambda, (r,-,-), \gamma) = (q_d,u)$
\end{enumerate}

If the current leaf has no right descendant and $M$ is in state $q_r$,
a right descendant is pushed (Transition~4). Again, state $q_p$
is entered. If the new leaf was pushed as right descendant, the tree pointer
is moved up while state $q_d$ is entered (Transition 5). State $q_d$
indicates that the current node has been processed entirely.

\begin{enumerate}\setcounter{enumi}{5}
\item
$\delta(q_d, \lambda, (l,*,*), \gamma) = (q_r, u)$
\smallskip
\item
$\delta(q_d, \lambda, (r,*,*), \gamma) = (q_d, u)$
\end{enumerate}

In state $q_d$, the tree pointer is moved to the ancestor. However,
if it comes to the ancestor from the left subtree, the right subtree is still
to be processed. In this case, Transition~6 sends the tree pointer to the
ancestor in state $q_r$. If the tree pointer comes to the ancestor from the
right subtree, the ancestor has entirely be processed and the tree pointer
is moved up in the appropriate state~$q_d$ (Transition~7).

\begin{enumerate}\setcounter{enumi}{7}
\item
$\delta(q_r, \lambda, (*,+,+), \gamma) = (q_l, d_r)$
\end{enumerate}

If there is a right descendant of the node visited in state $q_r$ then
the process is recursively applied to the right subtree by moving the 
tree pointer to the right descendant in state $q_l$ (Transition~8). 

The end of the phase that can uniquely be detected by $M$ when its tree 
pointer comes back to the root in state $q_d$ from the right.

Before we next turn to the extension of $M$, we consider the number of steps
taken to generate the complete binary trees. 
The total number of nodes in such a tree of level $\ell\geq 1$ is $2^\ell-1$.
Since all nodes except the root are connected by exactly one edge, 
the number of edges is $2^\ell-2$.
In order to increase the level of the tree-storage from $\ell$ to~\mbox{$\ell+1$,}
the tree pointer takes a tour through the tree as for a depth-first traversal.
So, every edge is moved along twice. In addition, each of the $2^\ell$ new nodes
is connected whereby for each new node the connecting (new) edge is also
moved along twice. In total, we obtain $2(2^\ell-2+2^\ell)= 2^{\ell+2}-4$
moves to increase the level.
Summing up the moves yields the number of moves taken by $M$ to increase 
the level of the tree-storage from initially $1$ to $\ell$ as
$$
\sum_{i=1}^{\ell-1} 2^{i+2}-4
=
-4(\ell-1)+ 2^{\ell+2}-8 = 2^{\ell+2}-4\ell-4.
$$

Now, the construction of $M$ is completed as follows.
Initially, $M$ performs~$8$ moves without any operation on the tree-storage.
That is, the tree pointer stays at the root. This can be realized by
additional states. Next, $M$ starts to run through the phases described above,
where at the end of phase $\ell-1$ the tree-storage forms a complete 
binary tree of level $\ell$. Before each phase,~$M$ performs
additionally $4$ moves without any operation on the tree-storage, respectively.

Finally, it remains to be described how the input is read and possibly
accepted. We let~$M$ read an input symbol at every move. An input
word is accepted if and only if its length is 
$2^0$, $2^1$, $2^2$, $2^3$, or
if~$M$ reads the last input symbol exactly
at the end of some phase.

In order to give evidence that $M$ works correctly, assume that the input
length is $2^x$, for some $x\geq 4$. Then $M$ starts to generate a
tree-storage that forms a complete binary tree of level $x-2$. 
The generation takes
$2^{x}-4(x-2)-4$ moves plus the initial delay of $8$ moves plus
the delay of totally $4(x-3)$ moves before each phase.
Altogether, this makes $2^{x}$ moves.
Since $M$ reads one input symbol at every move, it reads exactly $2^x$ symbols
and works in real time.
\eoe
\end{example}

\section{Computational Capacity}\label{sec:comp-cap}

This section is devoted to compare the computational capacity
of real-time deterministic tree-walking-storage automata with
some classical types of acceptors. On the bottom of the automata
hierarchy there are finite state automata characterizing the family
of regular languages $\reg$. Trivially, we have the inclusion 
$\reg \subset \rtf(\twsdnea)$ whose properness follows
from Example~\ref{exa:a2n}.

On the other end, we consider the deterministic linear bounded
automata that are characterizing the family of deterministic
context-sensitive languages $\dcsl$, that is, the complexity class
$\textsf{DSPACE}(n)$. In a real-time computation of some
$\twsda$, the tree-storage can grow not beyond $n+1$ nodes, where
$n$ is the length of the input. Since a binary tree with $n$ nodes
can be encoded with $O(n)$ bits, the tree-storage can be simulated
in deterministic space $n$. Therefore, a real-time $\twsda$ can be
simulated by a deterministic linear bounded automaton and we obtain
the inclusion $\rtf(\twsda)\subseteq \dcsl$.
 
We continue the investigation by showing that the possibility to create
tree-storages of certain types in real time can be utilized to
accept further, even unary, languages by
real-time deterministic, even non-erasing, tree-walking-storage automata.
To this end, we make the construction of Example~\ref{exa:a2n} more
involved and consider the non-semilinear unary language of the words whose lengths
are double Fibonacci numbers.

The Fibonacci numbers form a sequence in which each number is the sum of the
two preceding ones. The sequence starts from 1 and 1 (sometimes in the
literature it starts from~0 and 1). A prefix of the sequence is 
$1, 1, 2, 3, 5, 8, 13, 21, 34, 55, 89$. Correspondingly, we are speaking
of the $i$th Fibonacci number $f_i$, where $i$ is the position in the sequence
starting from $1$. So, for example, $f_6$ is the number~$8$.
We are going to prove that the language
$L_\subtext{fib}=\{\,a^{2n} \mid n \mbox{ is a Fibonacci number}\,\}$
is accepted by some $\twsdnea$ in real time 
by showing that a $\twsdnea$ can successively 
create tree-storages that are Fibonacci trees.
\emph{Fibonacci trees} are recursively defined as follows.
The Fibonacci tree $F_0$ of level $0$ is the empty tree.
The Fibonacci tree~$F_1$ of level $1$ is the tree
that consists of one node only.
The Fibonacci tree $F_\ell$ of level $\ell\geq 2$ consists of the root
whose left subtree is a Fibonacci tree of level $\ell-1$ and whose
right subtree is a Fibonacci tree of level $\ell-2$
(see Figure~\ref{fig:fib-tree}). For our purposes, the number of nodes
of a Fibonacci tree is important. It is well known that the number of nodes
of Fibonacci tree $F_\ell$, for $\ell\geq 2$, is
$\nu_\ell=\nu_{\ell-1}+\nu_{\ell-2}+1$.
In other words, we obtain $\nu_\ell=f_{\ell+2}-1$.

\begin{figure}
\centering
    \includegraphics[width=.9\textwidth]{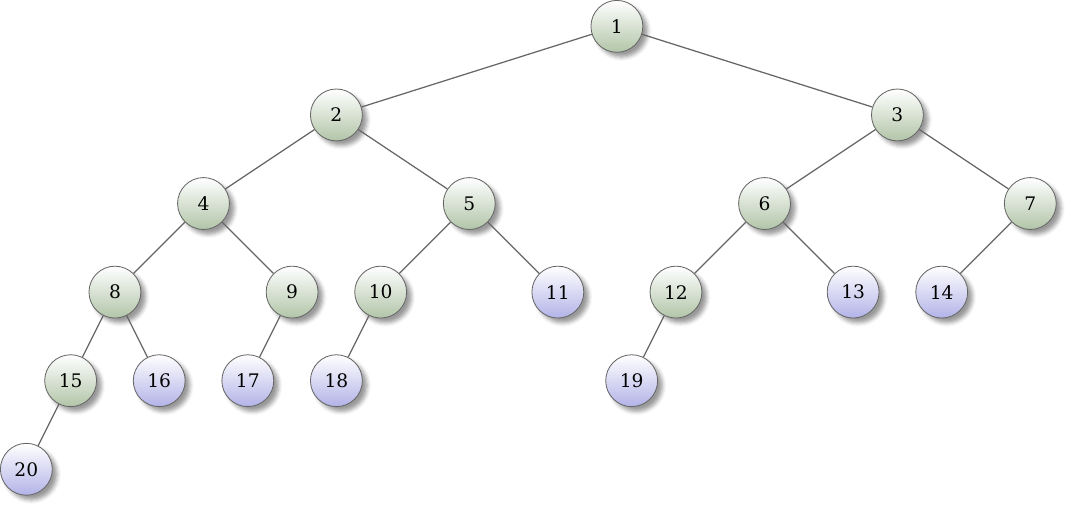}
    \caption{A Fibonacci tree of level $6$. Removing the blue nodes (the
      leaves) yields a Fibonacci tree of level $5$.}
    \label{fig:fib-tree}
\end{figure}

\begin{theorem}\label{theo:fib-in-twsdnea}
The language $L_\subtext{fib}$ is accepted by some $\twsdnea$ in real time.
\end{theorem}

\begin{proof}
We proceed as in Example~\ref{exa:a2n} and construct a 
$\twsdnea$ $M=\langle Q, \{a\}, \{\bullet\}, \delta, q_l, \rightend, \bot,
F\rangle$ with state set
$Q=\{q_l,q_p,q_d,q_r\}$
that runs in phases. Again, at the outset of the computation the tree-storage of $M$
forms a Fibonacci tree of level~$1$. After the $(\ell-1)$th phase, the
tree-storage of $M$ forms a Fibonacci tree of level $\ell$.
At the beginning and at the end of each phase the tree pointer
is at the root of the tree-storage. Again, we first construct $M$ such that
it works on empty input and extend it later.

So, assume that the tree-storage of $M$ forms a Fibonacci tree of level
$\ell\geq 1$ and that its tree pointer is at the root.
According to the recursive definition of Fibonacci trees, $M$ will increase
the levels of the subtrees of every node by one in a bottom-up fashion.
To this end, first state $q_l$ is used to move the tree pointer as far as
possible to the left. The leaf reached is the first node to be dealt with.
In particular, this leaf gets a left descendant. See Figure~\ref{fig:fib-tree}
for an example, where the Fibonacci tree of level 5 depicted by the
green nodes is extended to the entire Fibonacci tree of level 6 by
adding the blue nodes.

Let a star $*$ as component of the type of the current node in the
tree-walking-storage of~$M$ denote an arbitrary entry 
and $\gamma\in \Gamma\cup\{\bot\}$.
We set:
\begin{enumerate}
\item
$\delta(q_l, \lambda, (*,+,*), \gamma) = (q_l, d_l)$
\smallskip
\item
$\delta(q_l, \lambda, (*,-,-), \gamma) = (q_p,\pdpush(\bullet,l))$
\smallskip
\item
$\delta(q_p, \lambda, (*,*,*), \gamma) = (q_d,u)$
\end{enumerate}

Essentially, the meaning of the states are as in Example~\ref{exa:a2n}.
State $q_p$ indicates that the last tree operation
was a push, and the meaning of state $q_d$ is to indicate that the
current node has entirely be processed and that its ancestor is the next
node to consider.
So far, in Figure~\ref{fig:fib-tree} node 20 has been
pushed and the tree pointer is back at node 15 in state $q_d$.

\begin{enumerate}\setcounter{enumi}{3}
\item
$\delta(q_d, \lambda, (l,*,*), \gamma) = (q_r, u)$
\smallskip
\item
$\delta(q_d, \lambda, (r,*,*), \gamma) = (q_d, u)$
\end{enumerate}

Node 15 has entirely be processed, since it got a new left subtree of level~1
and, thus, stick with a right subtree of level 0 (the empty tree).
By Transitions~4 and~5 the tree pointer is moved to the ancestor. However,
if it comes to the ancestor from the left subtree, the right subtree is still
to be processed. In this case, Transition~4 sends the tree pointer to the
ancestor in state $q_r$. If the tree pointer comes to the ancestor from the
right subtree, the ancestor has entirely be processed and the tree pointer
is moved up in the appropriate state~$q_d$ (Transition~5).

\begin{enumerate}\setcounter{enumi}{5}
\item
$\delta(q_r, \lambda, (*,+,-), \gamma) = (q_p,\pdpush(\bullet,r))$
\smallskip
\item
$\delta(q_r, \lambda, (*,+,+), \gamma) = (q_l, d_r)$
\end{enumerate}

If there is a right descendant of the node visited in state $q_r$ then
the process is recursively applied to the right subtree by moving the 
tree pointer to the right descendant in state $q_l$ (Transition~7). Otherwise, if there
is no right descendant of the node visited in state $q_r$ then this
empty right subtree has to be replaced by a subtree of level 1. This is
simply done by pushing a single node (Transition~6).
In Figure~\ref{fig:fib-tree}, node~16 has been pushed as right descendant of
node 8. 
Then, after the next few steps, node 8 has entirely processed and
node 4 is reached in state~$q_r$. Continuing, this process will end
when node 3 has entirely been processed and the root is reached from
the right subtree in state $q_d$. This is the end of the phase that can
uniquely be detected by $M$ when its tree pointer comes back to the root
from the right.

Before we next turn to the extension of $M$, we consider the number of steps
taken to generate the Fibonacci tree. 

To this end, let $\ell\geq 1$ and recall that the number of nodes
of Fibonacci tree~$F_\ell$ is $\nu_\ell=f_{\ell+2}-1$. Since all
nodes except the root are connected by exactly one edge, 
the number of edges of Fibonacci tree $F_\ell$ is $\kappa_\ell=f_{\ell+2}-2$.
We derive that the number of nodes of $F_{\ell+1}$ is 
$\nu_{\ell+1}=f_{\ell+3}-1= f_{\ell+2}-1+f_{\ell+1}$ and the number of its
edges is $\kappa_{\ell+1}= f_{\ell+2}-2+f_{\ell+1}$.
In order to increase the level of the tree-storage from $\ell$ to~\mbox{$\ell+1$,}
the tree pointer takes a tour through the tree as for a depth-first traversal.
So, every edge of $F_\ell$ is moved along twice. In addition, each new node
is connected whereby for each new node the connecting (new) edge is also
moved along twice. In total, we obtain $2(f_{\ell+2}-2)$ plus $2f_{\ell+1}$
moves, that is, $2f_{\ell+3}-4$ moves. 
Summing up the moves yields the number of moves taken by $M$ to increase 
the level of the tree-storage from initially $1$ to $\ell$ as
\begin{multline*}
\sum_{i=1}^{\ell-1} 2 f_{i+3}-4
=
-4(\ell-1)+ 2\sum_{i=1}^{\ell-1} f_{i+3}
=
-4(\ell-1)-8  + 2\sum_{i=1}^{\ell+2} f_{i}\\
=
2(f_{\ell+4}-1) -4\ell-4
=
2f_{\ell+4} -4\ell-6,
\end{multline*}
since, in general,
$\sum_{i=1}^{\ell} f_{i} = f_{\ell+2}-1$.

Now, the construction of $M$ is completed as follows.
Initially, $M$ performs~$6$ moves without any operation on the tree-storage.
That is, the tree pointer stays at the root. This can be realized by
additional states. Next, $M$ starts to run through the phases described above,
where at the end of phase $\ell$ the tree-storage forms a Fibonacci tree
of level $\ell+1$. Before the first and after each phase,~$M$ performs
additionally $4$ moves without any operation on the tree-storage, respectively.

Finally, it remains to be described how the input is read and possibly
accepted. We let~$M$ read an input symbol at every move. An input
word is accepted if and only if its length is $2f_1$, $2f_2$, $2f_3$, $2f_4$, or
if~$M$ reads the last input symbol exactly
at the end of some phase.
In order to give evidence that $M$ works correctly, assume that the input
length is $2f_x$, for some $x\geq 5$. Then $M$ starts to generate a
tree-storage
that forms a Fibonacci tree of level $x-4$. The generation takes
$2f_{x} -4(x-4)-6$ moves plus the initial delay of $6$ moves plus
the delay of totally $4(x-4)$ moves before the first and after each phase.
Altogether, this makes 
$2f_{x} -4(x-4)-6+6+4(x-4)=2f_x$ moves. Since $M$ reads one input symbol at
every move, it reads exactly $2f_x$ symbols. Clearly, $M$ works in
real time.
\end{proof}

Now we turn to a technique for disproving that languages are 
accepted. In general, the method is based on equivalence
classes which are induced by formal languages. If some language
induces a number of equivalence classes which exceeds the number 
of classes distinguishable by a certain device, then the
language is not accepted by that device.
First we give the definition of an equivalence relation
which applies to real-time $\twsda$s.

Let $L\subseteq \Sigma^*$ be a language and $\ell\geq 1$ be an integer constant. 
Two words $w\in \Sigma^*$ and $w' \in \Sigma^*$ are 
\emph{$\ell$-equivalent with respect to $L$}
if and only if 
$
wu\in L \iff w'u\in L
$
for all $u\in \Sigma^*$, $|u|\leq \ell$.
The number of $\ell$-equivalence classes with respect to $L$ is 
denoted by $E(L,\ell)$.

\begin{lemma}\label{lem:rttwsca-equivalence-classes}
Let $L\subseteq \Sigma^*$ be a language accepted by some $\twsda$ in real
time. Then there exists a constant $p\geq 1$ such that
$E(L,\ell)\leq 2^{p\cdot 2^{\ell}}$.
\end{lemma}

\begin{proof}
The number of different binary trees with $n$ nodes is known to be the $n$th
Catalan number $C_n$. We have $C_0=1$ and
$C_{n+1}=\frac{4n+2}{n+2}C_n$ (see, for example,~\cite{stanley:2015:cn:book}).
So, we obtain $C_n\leq 4^n$, which is a rough but for our purposes good enough
estimation.

Now, let $M$ be a real-time $\twsda$
with state set $Q$ and tree symbols $\Gamma$.
In order to determine an upper bound for the number of $\ell$-equivalence 
classes with respect to $L(M)$, we consider the possible 
configurations of $M$ after reading all but $\ell$ input symbols. 
The remaining computation depends on the last $\ell$ input symbols,
the current state of $M$, the current $\Gamma$-tree as well as the current
tree pointer $P$. Since~$M$ works in real time, in its last at most $\ell+1$
steps it can only access at most $\ell+1$ tree nodes, starting with the
current node. These may be located in the upper $\ell$ levels of 
the tree rooted in the current node, or at the upper $\ell-1$ levels of 
the tree rooted in the ancestor of the current node, etc. So, there are
no more than $2^{\ell+1}-1+2^{\ell-1}+2^{\ell-2}+\cdots +2^{0}\leq 2^{\ell+2}$
nodes that can be accessed. Though the corresponding part of the tree can have
certain structures only, we consider all non-isomorphic binary trees with
$2^{\ell+2}$ nodes. Each node may be labeled by a symbol of $\Gamma$ or by
$\bot$.
Together, there are at most
$$
|Q| \cdot C_{2^{\ell+2}} \cdot (|\Gamma|+1)^{2^{\ell+2}}
\leq 
2^{\log(|Q|)+ 2\cdot 2^{\ell+2} + \log(|\Gamma|+1) \cdot 2^{\ell+2}}
=
2^{\log(|Q|)+ 4\cdot (2 + \log(|\Gamma|+1)) \cdot 2^{\ell}}
$$
different possibilities. Setting $p=\log(|Q|)+ 4(2 + \log(|\Gamma|+1))$, we
derive
$$
2^{\log(|Q|)+ 4\cdot (2 + \log(|\Gamma|+1)) \cdot 2^{\ell}}\leq
2^{p\cdot 2^{\ell}}.
$$

Since the number of equivalence classes is not
affected by the last~$\ell$ input symbols, there are at 
most~$2^{p\cdot 2^{\ell}}$ equivalence classes.
\end{proof}

Next, we turn to apply Lemma~\ref{lem:rttwsca-equivalence-classes}
to show that there is a context-free language
which is not accepted by any $\twsda$ in real time.
To this end, we consider the homomorphism 
$h\colon \{\alpha_0,\alpha_1,\alpha_2,\alpha_3\}^*\to \{a,b\}^*$ defined as
$h(\alpha_0)=aa$, $h(\alpha_1)=ab$, $h(\alpha_2)=ba$, $h(\alpha_3)=bb$, 
and the witness language
$$
L_h=\{\, x_1 \dollar x_2\dollar \cdots\dollar x_k \border y \mid 
k\geq 0, x_i \in \{a,b\}^*, 1\leq i\leq k, \mbox{ and there exists } j \text{
  such that } x_j^R=h(y)\,\}.
$$

\begin{theorem}\label{theo:lindet-notin-rttwsda}
The language $L_h$
is not accepted by any $\twsda$ in real time.
\end{theorem}

\begin{proof}
We consider some integer constant $\ell\geq 1$ and show that
$E(L_h,\ell)$ exceeds the number of equivalence classes 
distinguishable by any real-time $\twsda$.
To this end, let $L_h^{(\ell)} \subset L_h$ be the language
of words from $L_h$ whose factors $x_i$, $1\leq i\leq k$,
all have length~$2\ell$.

There are $2^{2^{2\ell}}$ different subsets of $\{a,b\}^{2\ell}$.
For every subset $P=\{v_1,v_2,\dots, v_k\}\subseteq \{a,b\}^{2\ell}$,
we define a word  
$w_P=\dollar v_1 \dollar v_2\dollar \cdots\dollar v_k \border$. 
Now, let $P$ and $S$ be two different subsets. Then there is some
word $u\in \{a,b\}^{2\ell}$ such that $u$ belongs to the symmetric difference
of $P$ and $S$. Say, $u$ belongs to $P\setminus S$. Setting $\hat{u}=h^{-1}(u)$
We have $w_P \hat{u}^R \in L_h$ and $w_S \hat{u}^R \notin L_h$.
Therefore, language $L_h$ induces at 
least~$2^{2^{2\ell}}$ equivalence classes in $E(L_h,\ell)$.

On the other hand, if $L$ would be accepted by some real-time $\twsda$,
then, by Lemma~\ref{lem:rttwsca-equivalence-classes}, there is a 
constant $p\geq 1$ such that $E(L_h,\ell)\leq 2^{p\cdot 2^{\ell}}$. Since 
$L_h$ is infinite, we may choose $\ell$ large enough
such that $2^{2\ell} > p\cdot 2^{\ell}$.
\end{proof}

Since the language $L_h$ is context free and, on the other hand,
the non-semilinear unary language of Proposition~\ref{theo:fib-in-twsdnea}
belongs to $\rtf(\twsdnea)$, we have the following incomparabilities.

\begin{theorem}
The families $\rtf(\twsda)$ and $\rtf(\twsdnea)$ are both incomparable with
the family of context-free languages.
\end{theorem}

Next, we consider classical deterministic one-way stack automata.
It has been shown that the unary language 
$L_\subtext{cub}=\{\,a^{n^3}\mid n\geq 0\,\}$ is not accepted by any
$\dsa$~\cite{ogden:1969:itsl}.

\begin{proposition}\label{prop:a3-in-rttwsda}
The language $L_\subtext{cub}$ is accepted by some $\twsda$ in real time.
\end{proposition}

Proposition~\ref{prop:a3-in-rttwsda} and the result in~\cite{ogden:1969:itsl}
yield the following corollary.

\begin{corollary}
There is a language belonging to $\rtf(\twsda)$ 
that does not belong to $\lfam(\dsa)$.
\end{corollary}

\section{Basic Closure Properties}\label{sec:closure}

The goal of this section is to collect some basic closure properties
of the families $\rtf(\twsda)$ and $\rtf(\twsdnea)$.
In particular, we consider 
Boolean operations (complementation, union, intersection) and AFL operations 
(union, intersection with regular languages, homomorphism, inverse homomorphism, 
concatenation, iteration). 
The results are summarized in Table~\ref{tab:closure} at the end of the section.

It turns out that the two families in question have the same properties and,
in particular, share all but one of these closure properties with the important family
of deterministic context-free languages.

We start by mentioning the only two positive closure properties which 
more or less follow trivially from the definitions.

\begin{proposition}\label{prop:closure-compl}
The families $\rtf(\twsda)$ and $\rtf(\twsdnea)$ are closed under
complementation and intersection with regular languages.
\end{proposition}

\begin{proof}
For acceptance it is required that the tree-walking-storage automata halt
accepting after having read the input entirely. Due to the real-time
requirement the machines halt in any case. Should this happen somewhere
in the input, the remaining input can be read in an extra state. So,
interchanging accepting and non-accepting states is sufficient to 
accept the complement of a language.

For the intersection with regular languages, it is enough to simulate
a deterministic finite automaton in the states which is a standard
construction for automata.
\end{proof}

In order to prepare for further (non-)closure properties, we now tweak
the language $L_h$ of Section~\ref{sec:comp-cap} and define
\begin{multline*}
L_p=\{\,x_1 \dollar^{|x_1|} x_2\dollar^{|x_2|} \cdots x_k\dollar^{|x_k|} \border y \mid 
k\geq 0, x_i \in \{a,b\}^*, 1\leq i\leq k,\\ \text{ no } x_i \text{ is proper prefix 
of } x_j, \text{ for } 1\leq j< i, \text{ and there exists } m \text{
  such that } x_m = y\,\}.
\end{multline*}

These little changes have a big impact. The language becomes now
real-time acceptable by some $\twsdnea$ $M$.
The basic idea of the construction of $M$ is that it can accept $L_p$ by
building a trie from $x_1,x_2,\dots, x_k$, observing that the $\dollar$
padding allows it to return to the root between each part, and then on 
encountering $\border$ it matches $y$ to the trie.

\begin{theorem}\label{theo:tweaked-dict-in-rttwsdnea}
The language $L_p$ is accepted by some $\twsdnea$ in real time.
\end{theorem}

The construction in the proof of Theorem~\ref{theo:tweaked-dict-in-rttwsdnea}
can straightforwardly be extended to show that the following language
$\hat{L}_p$ is also accepted by some $\twsdnea$ in real time.
\begin{multline*}
\hat{L}_p=\{\,x_1 \dollar^{|x_1|} x_2\dollar^{|x_2|} \cdots x_k\dollar^{|x_k|}
\cent z \border_1 y \mid 
k\geq 0, x_i \in \{a,b\}^*, 1\leq i\leq k, z\in\{a,b,\dollar\}^*\\ 
\text{ no } x_i \text{ is proper prefix 
of } x_j, \text{ for } 1\leq j< i, \text{ and there exists } m \text{
  such that } x_m = y\,\}.
\end{multline*}

The language
$$
\hat{L}_\subtext{mi} = \{\, x \cent v \dollar v^R \border_2 \mid 
x\in \{a,b,\dollar\}^*, v\in \{a,b\}^*\,\}
$$
is accepted by some deterministic pushdown automaton in real time. Therefore,
it is accepted by some real-time $\twsdnea$ as well.

The proof of the next Proposition first shows the non-closure under union. 
Then the non-closure under intersection 
follows from the closure under complementation by De Morgan's law.
A witness for the non-closure under union is 
$L=\hat{L}_p\cup \hat{L}_\subtext{mi}$.
No real-time $\twsda$ can accept $L$ as any deterministic automaton would have 
to represent a tree with arbitrary height representing a potential $v$ from 
$\hat{L}_\subtext{mi}$, which makes it impossible for it to reach whatever 
representation it has built of $x_1, x_2,\dots, x_k$
if it turns out to be trying to accept $\hat{L}_p$.

\begin{proposition}\label{prop:nonclosure-union}
The families $\rtf(\twsda)$ and $\rtf(\twsdnea)$ are neither closed under
union nor under intersection.
\end{proposition}

We turn to the catenation operations.

\begin{proposition}\label{prop:nonclosure-catenation}
The families $\rtf(\twsda)$ and $\rtf(\twsdnea)$ are neither closed under
concatenation nor under iteration.
\end{proposition}

\begin{proof}
To make the language $\hat{L}_p\cup \hat{L}_\subtext{mi}$ more manageable we 
add a hint to the left of the words. So, let $\bullet$ be a new symbol
and set $L_1= \bullet\hat{L}_p\cup \hat{L}_\subtext{mi}$. Since
$\hat{L}_p$ and $\hat{L}_\subtext{mi}$ do belong to $\lfam(\twsdnea)$,
$L_1$ is accepted by some real-time $\twsdnea$ as well.
The second language used here is the finite language $L_2=\{\bullet,
\bullet\bullet\}$ that certainly also belongs to $\lfam(\twsdnea)$.

We consider the concatenation $L_2\cdot L_1$ and assume that it belongs to
$\lfam(\twsda)$. Since $\lfam(\twsda)$ is closed under intersection with
regular languages, 
$(L_2\cdot L_1)\cap \bullet\bullet\{a,b,\dollar,\cent,\border_1,\border_2\}^*
= \bullet\bullet (\hat{L}_p\cup \hat{L}_\subtext{mi})$
belongs to $\lfam(\twsda)$. Since $\lfam(\twsda)$ is straightforwardly
closed under left quotient by a singleton, we obtain
$\hat{L}_p\cup \hat{L}_\subtext{mi} \in \lfam(\twsda)$, a contradiction.

\begin{sloppypar}
The non-closure under iteration follows similarly. Since $L_2$ is regular,
we derive that \mbox{$L_1\cup L_2\in \lfam(\twsda)$.} However,
$(L_1\cup L_2)^* \cap
\bullet\bullet\{a,b,\dollar,\cent,\border_1,\border_2\}^+$
equals again $\bullet\bullet (\hat{L}_p\cup \hat{L}_\subtext{mi})$. So,
as for the concatenation we obtain a contradiction to the assumption that
$\lfam(\twsda)$ is closed under iteration.
\end{sloppypar}
\end{proof}

\begin{proposition}\label{prop:nonclosure-hom}
The families $\rtf(\twsda)$ and $\rtf(\twsdnea)$ are not closed under
length-preserving homomorphisms.
\end{proposition}

\begin{proof}
The idea to show the non-closure is first to provide some hint that allows
a language to be accepted, and then to make the hint worthless by applying
a homomorphism.

So, let us provide a hint that makes the language $\hat{L}_p\cup
\hat{L}_\subtext{mi}$ acceptable by some real-time $\twsdnea$.
We use two new symbols $\bullet_1$ and $\bullet_2$
and set
$L= \bullet_1\hat{L}_p\cup \bullet_2\hat{L}_\subtext{mi}$. 
In this way, $L$ belongs to $\rtf(\twsdnea)$.
However applying the homomorphism
$h\colon \{a,b,\dollar,\cent,\border_1,\border_2,\bullet_1,\bullet_2\}^*
\to \{a,b,\dollar,\cent,\border_1,\border_2,\bullet\}^*$, that maps
$\bullet_1$ and $\bullet_2$ to $\bullet$ and all other symbols to itself,
to language $L$ yields
$h(L)= \bullet(\hat{L}_p\cup \hat{L}_\subtext{mi})$ which does not belong to
 $\rtf(\twsda)$.
\end{proof}

\begin{proposition}\label{prop:nonclosure-invhom}
The families $\rtf(\twsda)$ and $\rtf(\twsdnea)$ are not closed under
inverse homomorphisms.
\end{proposition}

\begin{proof}
Previously, we have taken the language $L_h\notin \rtf(\twsda)$ and tweaked
it to $L_p\in \rtf(\twsdnea)$.
Now we merge both languages to
\begin{multline*}
\tilde{L}_h=\{\,x_1 \dollar^{|x_1|} x_2\dollar^{|x_2|} \cdots
x_k\dollar^{|x_k|} \border y \mid 
k\geq 0, x_i \in \{a,b\}^*, 1\leq i\leq k,\\ \text{ no } x_i \text{ is proper prefix 
of } x_j, \text{ for } 1\leq j< i, \text{ and there exists } m \text{
  such that } x_m = h(y)\,\},
\end{multline*}
where 
$h\colon \{\alpha_0,\alpha_1,\alpha_2,\alpha_3\}^*\to \{a,b\}^*$ is defined as
$h(\alpha_0)=aa$, $h(\alpha_1)=ab$, $h(\alpha_2)=ba$, and $h(\alpha_3)=bb$. 
The main ingredients to show that $L_h\notin \rtf(\twsda)$
(Theorem~\ref{theo:lindet-notin-rttwsda}) are kept such that
$\tilde{L}_h\notin \rtf(\twsda)$ immediately follows. 

Similarly, if we require that $y\in\{a',b'\}$ has to match a factor $x_i$
after being unprimed then the corresponding language
\begin{multline*}
\tilde{L}_p=\{\,x_1 \dollar^{|x_1|} x_2\dollar^{|x_2|} \cdots x_k\dollar^{|x_k|} \border y \mid 
k\geq 0, x_i \in \{a,b\}^*, 1\leq i\leq k,\\ \text{ no } x_i \text{ is proper prefix 
of } x_j, \text{ for } 1\leq j< i, \text{ and there exists } m \text{
  such that } x_m = h_1(y)\,\}
\end{multline*}
where
$h_1\colon \{a',b'\}^*\to \{a,b\}^*$ is defined as
$h_1(a')=a$, and $h_1(b')=b$, still belongs to $\rtf(\twsdnea)$.

We define the homomorphism 
$h_2\colon
\{\alpha_0,\alpha_1,\alpha_2,\alpha_3,a,b,\dollar,\cent,\border_{1},\border_{2}\}^*\to 
\{a,b,\dollar,\cent,\border_1,\border_2,a',b'\}^*$ as\newline
\mbox{$h_2(\alpha_0)=a'a'$,} $h_2(\alpha_1)=a'b'$, $h_2(\alpha_2)=b'a'$,
$h_2(\alpha_3)=b'b'$,
and $h_2(x)=x$, for $x\in \{a,b,\dollar,\cent,\border_1,\border_2\}$.

So, we have
$h_2^{-1}(\tilde{L}_p)=\tilde{L}_h$ which implies the non-closure under
inverse homomorphisms.
\end{proof}

Finally, we consider the reversal.

\begin{proposition}\label{prop:nonclosure-reversal}
The families $\rtf(\twsda)$ and $\rtf(\twsdnea)$ are not closed under
reversal.
\end{proposition}

\begin{proof}
A witness for the non-closure under reversal is 
the language $L=\hat{L}_p\cup \hat{L}_\subtext{mi}$. By
Proposition~\ref{prop:nonclosure-union}, it is not accepted by
any real-time $\twsda$.

Concerning $L^R$, the first symbol of an input decides
to which language it still may belong. If the symbol is $\border_2$
the input may only belong to $\hat{L}^R_\subtext{mi}$. If it is from
$\{a,b,\border_2\}$ then the input may only belong to~$\hat{L}^R_p$.

The language $\hat{L}^R_\subtext{mi}$ is accepted by some real-time
deterministic pushdown automaton and, thus, by some real-time $\twsdnea$.
Furthermore, it is not hard to see that $\hat{L}^R_p$ belongs to
$\rtf(\twsdnea)$ as well. We conclude the non-closures under reversal.
\end{proof}

\begin{table}[!ht]
\begin{center}
\renewcommand{\arraystretch}{1.2}\setlength{\tabcolsep}{6pt}
\begin{tabular}{|c|c|c|c|c|c|c|c|c|c|}
\hline
  Family & $\overline{\phantom{aa}}$ & $\cup$ & $\cap$ & $\cap_\subtext{reg}$ & $\cdot$ & $*$ &
  $h_{\text{len.pres.}}$ & $h^{-1}$ & $R$\\
\hline\hline
$\rtf(\twsda)$   & \cyes & \cno & \cno & \cyes & \cno & \cno & \cno & \cno & \cno\\
$\rtf(\twsdnea)$   & \cyes & \cno & \cno & \cyes & \cno & \cno & \cno & \cno & \cno\\
$\dcfl$   & \cyes & \cno & \cno & \cyes & \cno & \cno & \cno & \cyes & \cno\\
\hline
\end{tabular}
\end{center}
\caption{Closure properties of the language families discussed. $\dcfl$
  denotes the family of deterministic context-free languages.}
\label{tab:closure}
\end{table}

\section{Future Work}

We made some first steps to investigate deterministic real-time 
tree-walking-storage automata. Several possible lines of future research 
may be tackled. First of all, it would be natural to consider the
\emph{nondeterministic} variants of the model.
Decision problems and their computational complexities are an untouched area.
Another question is how and to which extent the capacities and complexities 
are changing in case of a unary input alphabet and/or a unary set of tree
symbols (which lead to the notion of counters in the classical models).

\bibliographystyle{eptcs}

\end{document}